\newif\ifarxiv
\arxivfalse 

\ifarxiv
\documentclass[letterpaper]{article} 
\usepackage[square,numbers]{natbib} 
\usepackage{authblk}
\usepackage[margin=1in]{geometry} 
\else
\documentclass[acmconf]{acmart} 
\fi
\AtBeginDocument{%
  \providecommand\BibTeX{{%
    \normalfont B\kern-0.5em{\scshape i\kern-0.25em b}\kern-0.8em\TeX}}}

\copyrightyear{2022}
\acmYear{2022}
\setcopyright{acmcopyright}\acmConference[FAccT '22]{2022 ACM Conference on Fairness, Accountability, and Transparency}{June 21--24, 2022}{Seoul, Republic of Korea}
\acmBooktitle{2022 ACM Conference on Fairness, Accountability, and Transparency (FAccT '22), June 21--24, 2022, Seoul, Republic of Korea}
\acmPrice{15.00}
\acmDOI{10.1145/3531146.3533221}
\acmISBN{978-1-4503-9352-2/22/06}

\ifarxiv
\else

\acmConference[FAccT '22]{FAccT '22: ACM Conference on Fairness, Accountability and Transparency}{June 21--24, 2022}{Seoul, South Korea}
\acmBooktitle{FAccT '22: ACM Conference on Fairness, Accountability and Transparency,
  June 21--24, 2022, Seoul, South Korea}
\acmConference[2022 ACM Conference on Fairness, Accountability, and Transparency (FAccT '22)]{June 21--24, 2022}{\\Seoul, Republic of Korea}
\fi 

\newcommand{\p}[1]{\left( #1 \right)}

\newcommand{\cd}[0]{\cdot}

\usepackage{amsmath}               
\usepackage{amsthm}                
\usepackage{thmtools,thm-restate}
\usepackage[autostyle=true,german=quotes]{csquotes}
\usepackage{multirow}
\usepackage{color}

\usepackage{array}
\usepackage{mathtools}
\newcolumntype{C}[1]{>{\centering\let\newline\\\arraybackslash\hspace{0pt}}m{#1}}

\usepackage{subcaption}
\usepackage{hyperref}

\newtheorem{lemma}{Lemma}

\newtheorem{definition}{Definition}
\newtheorem{assumption}{Assumption}
\newtheorem{example}{Example}

\newcommand{\comb}[0]{\ensuremath{c}}
\newcommand{\alg}[0]{\ensuremath{a}}
\newcommand{\human}[0]{\ensuremath{h}}
\newcommand{\Alg}[0]{\ensuremath{A}}
\newcommand{\Human}[0]{\ensuremath{H}}

\newcommand{\select}[0]{\ensuremath{w_h}}
\newcommand{\nspace}[0]{\ensuremath{N}}
\newcommand{\prob}[0]{\ensuremath{p}}

\ifarxiv

\title{Human-Algorithm Collaboration: Achieving Complementarity and Avoiding Unfairness}
\author[1]{Kate Donahue\thanks{kdonahue@cs.cornell.edu}\thanks{Work done during internship at Amazon.}} 
\author[2]{Alexandra Chouldechova}
\author[3]{Krishnaram Kenthapadi\thanks{Work done while at Amazon.}}
\affil[1]{Department of Computer Science, Cornell University}
\affil[2]{Carnegie Mellon University and Amazon Web Services}
\affil[3]{Fiddler AI}
\date{}
\begin{document}
\else 
\begin{document}

\title{Human-Algorithm Collaboration: Achieving Complementarity\\ and Avoiding Unfairness}

\author{Kate Donahue} 
\authornote{Work conducted while on an internship at Amazon.}
\email{kdonahue@cs.cornell.edu}
\affiliation{%
  \institution{Cornell University} 
  \country{USA}
}
\author{Alexandra Chouldechova}
\affiliation{%
  \institution{Carnegie Mellon University and Amazon Web Services}
  \country{USA}
}
\author{Krishnaram Kenthapadi}
\authornote{Work done while at Amazon.}
\affiliation{%
  \institution{Fiddler AI} 
  \country{USA}
}
\renewcommand{\shortauthors}{}

\renewcommand{\shortauthors}{Donahue, Chouldechova, Kenthapadi} 

\fi 

\ifarxiv
\maketitle 
\else
\fi

\begin{abstract}
Much of machine learning research focuses on predictive accuracy: given a task, create a machine learning model (or algorithm) that maximizes accuracy.  In many settings, however, the final prediction or decision of a system is under the control of a human, who uses an algorithm's output along with their own personal expertise in order to produce a combined prediction. One ultimate goal of such collaborative systems is \textit{complementarity}: that is, to produce lower loss (equivalently, greater payoff or utility) than either the human or algorithm alone.  However, experimental results have shown that even in carefully-designed systems, complementary performance can be elusive. Our work provides three key contributions. First, we provide a theoretical framework for modeling simple human-algorithm systems and demonstrate that multiple prior analyses can be expressed within it. Next, we use this model to prove conditions where complementarity is impossible, and give constructive examples of where complementarity is achievable. Finally, we discuss the implications of our findings, especially with respect to the fairness of a classifier. In sum, these results deepen our understanding of key factors influencing the combined performance of human-algorithm systems, giving insight into how algorithmic tools can best be designed for collaborative environments. 
\end{abstract}

\ifarxiv
\else
\maketitle 
\fi

\section{Introduction}

\begin{figure}[t]
  \begin{minipage}[t]{0.27\textwidth}
  \centering
    \includegraphics[width=\textwidth]{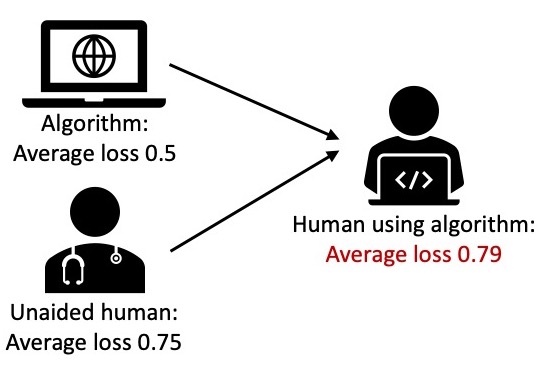}
     \subcaption{Scenario 1: Combined system has higher loss than either unaided human or algorithm. \label{fig:scenario1}}
  \end{minipage}
  \hspace{3em}
  \begin{minipage}[t]{0.27\textwidth}
  \centering
    \includegraphics[width=\textwidth]{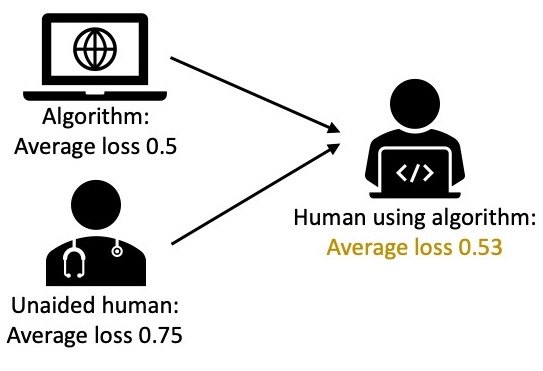}
    \subcaption{Scenario 2: Combined system has lower loss than unaided human, but higher loss than algorithm. \label{fig:scenario2}}
  \end{minipage}
  \hspace{3em}
    \begin{minipage}[t]{0.27\textwidth}
  \centering
    \includegraphics[width=\textwidth]{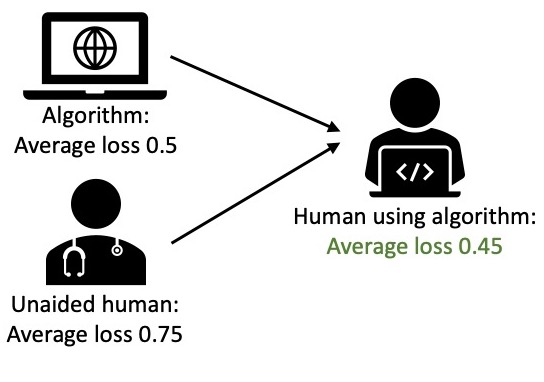}
    \subcaption{Scenario 3: Complementary performance: lower error than either unaided human or algorithm.\label{fig:scenario3}}
  \end{minipage}
  \caption{Three possible scenarios for human-algorithm collaboration, each with the same algorithmic and unaided human loss. However, the loss of the combined system (human using the algorithm) might vary substantially. Section \ref{sec:motivate} gives a more detailed analysis. }  \label{fig:threescenarios}
  \vspace{-0.6cm}
\end{figure}

Consider a prediction task where the goal is to take a set of features about the world as input and predict an outcome of interest. A typical machine learning approach to such a task is to attempt to select a model with low (generalization) loss for the problem at hand. If such a model is applied directly to the prediction task, it will minimize expected loss. 

However, this standard approach does not necessarily reflect the way that machine learning tools are actually implemented. Often, algorithmic predictions are presented to humans, who then make a final decision by additionally relying on their own expertise \citep{vodrahalli2021humans, bansal2021most, KerriganCombining, yin2019understanding}. For example, consider a doctor looking at a medical record and trying to make a determination of whether disease is present. An algorithmic prediction based on the record may be useful, but it almost certainly will not be the sole factor influencing the doctor's diagnosis. For example, the doctor may have access to different data, such as conversations with the patient. The doctor may also have access to different knowledge, such as distilled expertise from years of practice. The doctor's decision will be a function of the algorithm's prediction, as well as their own inherent belief. Note that the doctor's decision-making may be imperfect, such as relying on their own judgement when the algorithm may have better performance. A successful outcome occurs when the combined system (the doctor using algorithmic output) has low loss, not when the algorithm alone has low loss. Figure \ref{fig:threescenarios} illustrates three scenarios where a combined human-algorithm system could have differing levels of loss. 

In particular, one especially valuable goal is \emph{complementarity} (or \emph{complementary performance}). Complementarity (originally defined in \citet{bansal2021does}) is achieved whenever the combined human-algorithm system has strictly lower expected loss than either the human or the algorithm alone (Figure \ref{fig:scenario3}). Complementarity is not necessary for a combined system to be deemed successful: for example, a combined system that does better than the human alone, but not necessarily better than the algorithm alone, would still reflect an improvement from a human-alone status-quo. However, complementarity creates the strongest incentive for adoption of a combined human-algorithm system, which is why it is the focus of our analysis.

{\textbf{Contributions:}}
At a high level, we address the following problems: \emph{ (i) How do we formally and tractably model human-algorithm collaborative systems? (ii) When can human-algorithm collaborative systems produce higher accuracy than either the human or algorithm alone? (iii) What are the fairness implications of such collaborative systems? }

The contributions of this work are three-fold. First, in Section \ref{sec:modelassump}, we introduce a simple theoretical framework for analyzing human-algorithm collaboration, and demonstrate the richness of this framework by showing that it can encapsulate models from previous works analyzing human decision-making. In Section \ref{sec:motivate}, we provide a simple, concrete motivating example using this framework that illustrates the core results of this paper. 

Next, in Section \ref{sec:complementarity}, we use this approach to analyze complementarity. First, we present several impossibility results that characterize regimes in which human-algorithm collaboration can never achieve complementarity. We then give concrete conditions for when complementarity can be achieved. In particular, our results suggest that complementarity is easier to achieve when loss rates are highly variable: when the unaided human (or algorithm) has very low loss on some inputs and very high loss on inputs. Disparate levels of loss raises issues of fairness, which we turn to next.

In Section \ref{sec:fair} we conclude our analysis by examining the fairness impacts of complementarity. The variability in loss rates implied by our results has implications for fairness, since types of inputs with very high error rates may correspond to protected attributes, such as race, gender, or ethnicity. To investigate this concern, we propose and analyze three types of fairness relating to human-algorithm systems, giving conditions for when they can and cannot be achieved.  One of main results shows that when complementarity is achieved, at least one group does worse in the combined system than under the human-only status quo.  Additionally, we give a simple condition where the combined human-algorithm system will guarantee that loss disparity between different protected groups will not increase. 


\section{Related Work}

\subsection{Human-Algorithm Collaboration}
A series of papers have explored issues related to human-algorithm collaboration. For example, \citet{poursabzi2021manipulating} and \citet{yin2019understanding} analyze how explainability and accuracy, respectively, influence how humans use algorithmic predictions. Similarly, \citet{dietvorst2020people} hypothesize that humans may prefer algorithmic predictions that are more variable in their loss rates and \citet{dietvorst2018overcoming} suggests that allowing algorithmic predictions to be modified may make humans more likely to use them. 

Other papers center more on in-depth qualitative assessments of how professionals incorporate tailor-made algorithmic tools into their workflow. For example, \citet{lebovitz2021ai, lebovitz2020incorporate} studies how doctors in major US hospitals use AI predictions in their daily work.
Similarly, \citet{okolo2021cannot} studies how community healthcare workers in India believe AI tools could influence their work. Finally, \citet{yang2018investigating} studies how UX designers work with machine learning tools and the data scientists who create them. 

Some research teams that develop tools for human-in-the-loop settings have run experiments analyzing how their tools perform with human collaboration. For example, \citet{raghu2019algorithmic, beede2020human} both study how an AI tool for predicting diabetic retinopathy fits in with a broader ecosystem (human doctors and the overall healthcare system). Similarly, \citet{de2020case} studies how child welfare call screeners incorporate algorithmic predictions in their risk assessments.  \citet{tan2018investigating} studies how human and algorithmic distributions of loss rates differ for recividism predictions for the COMPAS dataset (but not in ways that allowed for complementary performance by combined systems). Similarly, \citet{geirhos2020beyond} compares the similarity (consistency) of loss in predictions made by humans and a deep learning algorithm.

Some computer science papers specifically analyze models of human-algorithm interaction, such as \citep{cowgill2020algorithmic, sayres2019using, srivastava2020empirical, bansal2021most, KerriganCombining, vodrahalli2021humans}. Of these, \citet{bansal2021does} is especially relevant because it is framed through the goal of complementarity. \citet{bansal2021most} also highlights the fact that the optimizing for the algorithm's error may not minimize the loss of the combined system. In Section \ref{sec:modelassump}, we show how human decision-making processes rules inspired by analyses \citet{bansal2021most} and \citet{vodrahalli2021humans} can be represented in our model. Some papers show how to build models optimized for a human-algorithm \emph{deferral} system, where the final decision is made by either the human or the algorithm \citep{de2020classification, okati2021differentiable, meresht2020learning}. \citet{straitouri2022provably} studies a variant of this problem for classification where the algorithm presents a subset of possible labels to the human, who selects the final decision from among them. Finally, \citet{cabitza2021studying} studies multiple methods of aggregating human and algorithmic predictions. 

\subsection{Fairness in Human-Algorithm Collaboration}
Some papers specifically consider the fairness implications of combined human-algorithm systems. For example, \citet{madras2017predict} studies fairness and accuracy in deferring to a human expert, while \citet{keswani2021towards, keswani2022designing} extends this analysis to deferring to multiple different human experts.  For example, \citet{gillis2021fairness} takes a theoretical approach towards modeling the human-algorithm system and gives conditions where adding a biased (unfair) human can change the fairness properties of the overall system. \citet{valera2018enhancing} studies a system with multiple biased \enquote{experts}, where assigning the correct expert to each task can improve accuracy while still satisfying fairness requirements. 
\subsection{Related Papers From Other Areas}
Finally, some papers in seemingly unrelated areas end up being relevant to our analysis. For example, ensemble learning studies how to incorporate predictions from multiple algorithms into a unified (more accurate) system \cite{kuncheva2014combining}. Ensemble learning differs from our analysis in that each expert (predictor) is assumed to be an algorithm, and predictions are assumed to be combined by some additional algorithm under our control (rather than a human decision-maker we cannot control). However, certain factors identified in the ensemble learning literature as affecting overall performance, such as diversity, are relevant for our analysis \cite{didaci2013diversity}. Additionally, multiple works study how fairness properties of predictors change when they are composed \cite{dwork2018fairness, dwork2020individual, wang2019practical}. These works are relevant to our analysis in Section \ref{sec:fair} of the fairness of a combined human-algorithm system, but differ somewhat from ours: in general, these other papers tend to study fairness of allocating or achieving some desired prediction, while our analysis describes fairness as equal loss across groups. Finally, \citet{meehl1954clinical} compares statistical and clinical methods of reasoning, a framing that parallels to our analysis of algorithmic versus human prediction methods. 

\section{Model and Assumptions}\label{sec:modelassump}

\subsection{Model}
Our model considers a prediction task: given some element $x \in \mathcal{X}$, make a prediction $y \in \mathcal{Y}$ that minimizes some loss function $\mathcal{L}$, with loss bounded $\geq 0$. This loss could reflect any error rates for any type of learning problem---for example, regression and classification tasks could both be represented by this loss function. We model the input space $\mathcal{X}$ as being made up of $\nspace$ discrete regimes: all inputs within the same regime are identical from the perspective of algorithmic and human loss. This is without loss of generality, given that $N$ could be arbitrarily large. We are \emph{not} assuming that either the human or algorithm has knowledge of these regimes, simply that they exist. We will denote the probability of seeing regime $i$ is given by $\prob_i$, with $\sum_{i =1}^{\nspace}\prob_i = 1$. 

The human-algorithm system consists of three components: 
\begin{enumerate}
    \item An \emph{algorithm}, which for each regime in the input space $x_i \in \mathcal{X}$ makes a prediction $\hat y_i^a$ with some loss rate $\alg_i$. The average loss is given by $\sum_{i =1}^{\nspace}\prob_i \cd \alg_i = \Alg$. We can write $\alg_i = \Alg + \delta_{ai}$, with  $\sum_{i=1}^{\nspace} \prob_i \cd \delta_{ai} = 0$. The term $\delta_{ai}$ represents how much $\alg_i$ varies (differs from the average loss $\Alg$). 
    \item A \emph{unaided human}, which similarly for each regime in the input space $x_i \in \mathcal{X}$ makes some prediction $\hat y_{i}^h$. The average loss of the human is given by $\sum_{i =1}^{\nspace}\prob_i \cd \human_i = \Human$. Similarly, we write write $\human_i = \Human + \delta_{hi}$, with  $\sum_{i=1}^{\nspace} \prob_i \cd \delta_{hi} = 0$. 
    \item Finally, some \emph{combiner} (a human using algorithmic input) $g(\hat y_{i}^a, \hat y_{i}^h)$, which takes predictions given by the algorithm and unaided human and returns a combined prediction, $\hat y_i^c$. The combining function reflects human decision-making: it could select the algorithm's prediction, the unaided human's prediction, or interpolate between the two of them. We could also view this as a (loss) \emph{combining function} $\comb(\alg_i, \human_i)$ that takes the algorithmic loss and human loss on a particular instance and returns some combined loss. 
\end{enumerate}

In general, we may not have control over all (or even any) of these components. For example, the combining function reflects human judgement, which typically can't be directly manipulated. A primary goal of our analyses is to determine when a human-algorithm system displays complementarity, defined in Definition \ref{def:comp} below. 

\begin{definition}[From \citet{bansal2021does}]\label{def:comp}
A human-algorithm system displays \textbf{complementary performance} when the combined system has (strictly) lower loss than either the human or algorithm: 
$$\sum_{i=1}^{\nspace}\prob_i \cd \comb(\alg_i, \human_i) < \min\p{\sum_{i =1}^{\nspace}\prob_i \cd \alg_i,\sum_{i =1}^{\nspace}\prob_i \cd \human_i}= \min(\Alg, \Human)$$
\end{definition}

\subsection{Assumptions}
The combining function models the key question in human-algorithm collaboration: how do humans incorporate algorithmic predictions with their own expertise? In this work, we will make two main assumptions about how such combination occurs. 
First, throughout this paper, we will find it useful to work in the space of combining \emph{losses}, rather than combining predictions. Specifically, we will use the $\comb(\alg_i, \human_i)$ loss combining function.  Assumption \ref{assump:lossonly} describes the assumption this implies.
\begin{assumption}\label{assump:lossonly}
The loss of a combined human-algorithm system can be modeled by a combining rule relying only on the loss rates of the unaided human and algorithm in a particular regime: $\comb(\alg_i, \human_i)$. That is, regimes with identical (unaided human, algorithm) pairs of loss rates are treated identically. 
\end{assumption}
This assumption reflects the case where the level of accuracy in the algorithm and (unaided) human is the only feature influencing the accuracy of the combined system. An example of a situation that might \emph{violate} this assumption is if regime 1 and 2 both have loss of 3\% for the unaided human and 5\% for the algorithm, but the human using the algorithm (combined system) has loss of 3\% for regime 1 and 4\% for regime 2. Considering the case where this assumption is relaxed could be an interesting avenue for future work: however, it would likely result in much more complicated analysis. 

Next, Assumption \ref{assump:bounded} below, describes the assumption that the combining rule's outputs are bounded. 
\begin{assumption}\label{assump:bounded}
For each regime, the loss of the combined system is bounded between the loss of the human and algorithm: 
$$ \min(\alg_i, \human_i) \leq \comb(\alg_i, \human_i) \leq \max(\alg_i, \human_i)$$
\end{assumption}

This assumption reflects a case where the combiner operates by interpolating between the predictions made by the human or algorithm. An example of a situation that might \emph{violate} this assumption is if the combined system has loss 3\% in a certain regime, where the human has loss 4\% and the algorithm has loss 6\% in that regime. A bounded combining rule makes modeling human-algorithm collaboration more realistic: complementarity is trivial to achieve if the combining rule's loss can be arbitrarily disconnected from the loss of the human and algorithm. 

It's worth considering why a system may satisfy bounded inputs (Assumption \ref{assump:bounded}) and still exhibit complementarity. Assumption \ref{assump:bounded} refers to bounds at the level of each \emph{regime}, while complementarity refers to \emph{overall average loss}. For example, Scenario 3 (Table \ref{tab:comp}) in Section \ref{sec:motivate} obeys the bound in Assumption \ref{assump:bounded}, and yet also achieves complementarity. 

\subsection{Weighting function}
In this work, we will find it helpful to think about the combined human-algorithm system as involving a \emph{weighting function} $0 \leq \select(\alg_i, \human_i)\leq 1$ controlling how much the human influences the final prediction:
\begin{equation}\label{eq:weighting}
\comb(\alg_i, \human_i) = (1-\select(\alg_i, \human_i)) \cd \alg_i + \select(\alg_i, \human_i) \cd \human_i 
\end{equation}
Lemma \ref{lem:equivalent}, below, shows that using a weighting function requires no new assumptions. 
\begin{restatable}{lemma}{equivalent}
\label{lem:equivalent}
Any combining rule relying only on loss rates (Assumption \ref{assump:lossonly}) with bounded output (Assumption \ref{assump:bounded}) can be written as a combining rule with a weighting function $0 \leq \select(\alg_i, \human_i)\leq 1$. 
\end{restatable}

One simple (ideal) combining function is given by Example \ref{ex:min}: it simply selects whichever of the unaided human or algorithm has lower loss. While this is the best possible combining function (given our assumptions), it is likely not a realistic model of how human decision-makers incorporate algorithmic advice. 

\begin{example}[Min]\label{ex:min}
The combining function becomes $\comb(\alg, \human) = \min(\alg, \human)$ is represented by the weighting function: 
$$\select(\alg, \human) = \begin{cases} 1 & \human \leq \alg \\
0 & \text{otherwise}
\end{cases}
$$   
\end{example}

While our framework is simple, it is also sufficiently flexible to capture models of human-algorithm collaboration studied in multiple previous papers (described in greater detail in Appendix \ref{app:modelingrules}). Examples \ref{ex:bansal} and \ref{ex:vodrahalli} demonstrate this in reference two particular models suggested by prior literature.
First, Example \ref{ex:bansal} selects whichever of the human or algorithm has lower loss rate with probability $p_s$. For high $p_s$, this reflects a decision-maker who accurately trusts whichever has lower loss.

\begin{example}[\citet{bansal2021most}]\label{ex:bansal}
The analysis in \citet{bansal2021most} suggests the weighting function:
$$\select(\alg, \human) = \begin{cases}
p_s & \human \leq \alg \\
1-p_s & \text{otherwise}
\end{cases}$$
\end{example}

Next, in Example \ref{ex:vodrahalli}, the decision-maker first decides whether to consider algorithmic advice at all: it does so only if the loss rate is $\epsilon$ lower than the human loss rate. Then, the decision-maker incorporates algorithmic advice with some probability $p_s(\cd)$ that is a function depending on the gap between human and algorithmic loss rates. 

\begin{example}[\citet{vodrahalli2021humans}]\label{ex:vodrahalli}
The two-stage model in \citet{vodrahalli2021humans} could be written as: 
$$\select(\alg, \human) = \begin{cases}
1 & \alg \geq \human - \epsilon \\
p_s(\human-\alg)  & \text{otherwise}
\end{cases}$$
\end{example}

\subsection{Research Ethics and Social Impact}
While our paper is primarily theoretical, its application area prompts a number of ethical considerations. For example, in this work, we are primarily concerned with building better prediction functions. In general, such functions could be used for positive means (helping a doctor correctly diagnose a disease) or negative ones (enabling the identification and repression of minority groups). Additionally, even if a function is being used for positive goals, it could be the case that factors besides the loss rate are ultimately more important. For example, it could be that the process of coming to a prediction, rather than the prediction itself, is more important. This is especially salient for our discussion of fairness, which assumes that the fairness of outcomes (of loss rate disparities) is the relevant factor to consider, rather than fairness of the prediction process. The issue of explanation of algorithmic predictions, which is orthogonal to our main analysis, could be relevant for this consideration.

\section{Motivating example}\label{sec:motivate}

\begin{table}[t]
  \begin{minipage}[t]{\textwidth}
  \centering 
\begin{tabular}{|c|c|c|c|c|c|}
\hline
                        & \textbf{(Unaided) human} & \textbf{Algorithm} &  \textbf{\begin{tabular}[c]{@{}c@{}}Combined\\ (human using algorithm)\end{tabular}}  & \textbf{\begin{tabular}[c]{@{}c@{}}Weight\\ (on unaided human)\end{tabular}} \\ \hline
\textbf{Regime 1} & 1            &  0.35              & 0.94      & 0.9                               \\ \hline
\textbf{Regime 2} & 0.5          & 0.65               & 0.64          & 0.1                            \\ \hline
\textbf{Average}  & 0.75             & 0.5                & 0.79         & 0.5                             \\ \hline
\end{tabular}
\caption{Scenario 1 Loss rates: An example of a combined human-algorithm system. There are two regimes, each making up equal proportions of the input space ($p_0 = p_1 = 0.5$). Note here that complementarity is \emph{not} satisfied: in fact, the combined system has higher loss than either the human or algorithm alone!}
\label{tab:nocomp}
  \end{minipage}
  \begin{minipage}[t]{\textwidth}
\centering 
\begin{tabular}{|c|c|c|c|c|c|}
\hline
                        & \textbf{(Unaided) human} & \textbf{Algorithm} & \textbf{\begin{tabular}[c]{@{}c@{}}Combined\\ (human using algorithm)\end{tabular}} &   \textbf{\begin{tabular}[c]{@{}c@{}}Weight\\ (on unaided human)\end{tabular}}\\ \hline
\textbf{Regime 1} & 1            &  0.35              & 0.51           & 0.25                        \\ \hline
\textbf{Regime 2} & 0.5          & 0.65             & 0.54            & 0.75                          \\ \hline
\textbf{Average}  & 0.75             & 0.5                & 0.53                  & 0.5                     \\ \hline
\end{tabular}
\caption{Scenario 2 Loss rates: A second example of a combined human-algorithm system, but with different loss distributions. Here, the combined system (human using algorithm) has average loss which is lower than the loss of the unaided human, but higher than the loss of the algorithm alone. }
\label{tab:onebetter}
  \end{minipage}
    \begin{minipage}[t]{\textwidth}
  \centering 
\begin{tabular}{|c|c|c|c|c|c|}
\hline
                        & \textbf{(Unaided) human} & \textbf{Algorithm} & \textbf{\begin{tabular}[c]{@{}c@{}}Combined\\ (human using algorithm)\end{tabular}}  &    \textbf{\begin{tabular}[c]{@{}c@{}}Weight\\ (on unaided human)\end{tabular}}\\ \hline
\textbf{Regime 1} & 1.15          & 0.2               & 0.44          & 0.25                              \\ \hline
\textbf{Regime 2} & 0.35            & 0.8                & 0.46            &0.75                              \\ \hline 
\textbf{Average}  & 0.75              & 0.5                & 0.45                  &0.5                        \\ \hline
\end{tabular}
\caption{Scenario 3 Loss rates: A third example of a combined human-algorithm system. Here, the combined system (human using algorithm) displays \emph{complementary performance}: its average loss of 0.45 is lower than the loss of either the unaided human (0.75) or algorithm alone (0.5). }
\vspace{-0.3cm}
\label{tab:comp}
  \end{minipage}
  \caption{Three possible scenarios for human-algorithm collaboration. In each, the algorithm and unaided human have the same average loss. However, the loss of the human using the algorithm varies. Figure \ref{fig:threescenarios} gives a visual description of these scenarios.}  \label{threescenarios}
  \vspace{-0.9cm}
\end{table}

To further motivate the analysis, let us revisit the medical application from the introduction (Figure \ref{fig:threescenarios}) and introduce further specifics\ifarxiv\footnote{Code for all examples and figures in this paper are available at \url{https://github.com/kpdonahue/human_algorithm_collaboration}}\fi.  Consider the medical prediction task of using information from a patient's medical record to predict disease severity (on a scale from 0 to 5, as in \cite{raghu2019algorithmic}). As illustrated in Figure \ref{fig:threescenarios}, we will assume that doctors relying on their medical training (unaided humans) have an average loss rate of 0.75: they are off by 0.75 grades, on average. A data science team has created a machine learning algorithm that has average loss of 0.5.  

Even though the algorithm has lower loss, doctors won't simply rubber-stamp algorithmic suggestions. Because they have specialized training and access to additional information (such as conversations with patients), a doctor might reasonably incorporate algorithm advice only partially, or only sometimes. However, this leaves open a crucial question: what is the combined human-algorithm loss? That is, what is the average loss once doctors start incorporating the new machine-learning algorithm into their decision-making process?
\subsection{Three scenarios}
To help build intuition before delving into our formal theoretical results, we will consider three different example scenarios for how a human-algorithm collaborative system might look like, given in Tables \ref{tab:nocomp}, \ref{tab:onebetter}, \ref{tab:comp}. Each scenario has the same average loss for the unaided human and for the algorithm (0.75 and 0.5, respectively). However, each scenario differs in 1) the way (unaided) human and algorithm loss is distributed and 2) the way the human combines algorithmic advice. Specifically, these three scenarios illustrate the simplified case where patient records (\enquote{regimes}) come in one of two types: regime 1 and regime 2, each of which makes up 50\% of the total input space. These regimes might differ in multiple ways: say disease progression, data quality, or patient characteristics. We will treat observations from within the same regime as identical: the doctor or algorithm or combined system each has uniform loss for every record within the same regime. (In later sections, we will relax this consider arbitrary numbers of regimes, reflecting arbitrarily complex distributions of loss). 

\textbf{Scenario 1 (combined loss higher than unaided human or algorithm):} For example, in Table \ref{tab:nocomp}, the unaided human has loss 1 on instances of type 1, but loss 0.5 for instances of type 2. The algorithm's loss rate distribution differs: 0.35 for type 1, and 0.65 for type 2. Finally, the combined loss (loss of the human using algorithmic input) for a particular regime is a function of the loss of the unaided human and the algorithm: for this example, it's 0.94 in regime 1, and 0.64 in regime 2. Note that this results in an average loss of 0.79---the human using the algorithm has a strictly \emph{greater} loss rate than either the unaided human or the algorithm! This unfortunate case could result from inappropriately relying on the algorithm: for example, the doctor might mistakenly incorporate algorithmic advice more frequently in regimes where it happens to have higher loss.  The fourth column concretely reflects this cause: it calculates the weighting function for each regime (reliance on the unaided human, as opposed to the algorithm). In this case, the regime 1 weighting function is 0.9, meaning the combined system is relying heavily on the unaided human for this instance, even though it has higher loss than the algorithm.  Similarly, in regime 2, the weighting function is 0.1, indicating that the combined system is relying more on the algorithm, even though (for this regime), the algorithm has higher loss than the human. This inappropriate reliance explains why the combined system has higher loss than either the unaided human or algorithm alone. 

\textbf{Scenario 2 (combined system lower error than human, higher error than algorithm):} Table \ref{tab:onebetter} presents a slightly more optimistic Scenario 2. Here, \emph{the distribution of loss rates for the unaided human and the algorithm are the same as in Scenario 1, but the combined loss differs}: the average loss of the human using the algorithm is 0.53---lower than the loss of the unaided human (0.75), but higher than the loss of the algorithm alone (0.5). This reflects a common scenario in human-algorithm collaboration: the combined system ends up improving over the human alone, but still falls short of the loss rate achievable by the algorithm alone \cite{bansal2021does}. In this case, the reason is because \emph{the combined system is doing a better job than in Scenario 1 of appropriately relying on the unaided human or algorithm}. Note that the weighting function is lower in regime 1 (when the unaided human has higher loss) and higher in regime 2 (where the unaided human has lower loss). This more appropriate reliance explains why the combined system has better performance than Scenario 1. 

\textbf{Scenario 3 (complementarity: lower than unaided human or algorithm):} Finally, Table \ref{tab:comp} presents Scenario 3. In this case, \emph{we keep the weighting function the same as in Scenario 2}, as well as the average loss of the unaided human and the algorithm. However, \emph{we change the distribution of loss rates across the regimes: we make them more variable}. For example, the unaided human now has loss 1.15 in regime 1 and 0.35 in regime 2, while the algorithm has loss 0.2 in regime 1 and 0.8 in regime 2. While average (unaided) human and algorithmic loss are the same as in Scenarios 1 and 2, the combined human-algorithmic system ends up having average loss of 0.45, which strictly lower than either the unaided human or algorithm alone.  Complementarity arises here because diversity of errors and appropriate reliance when algorithmic error is lower than the human’s.  As we show in Section \ref{sec:complementarity}, these conditions are necessary for complementarity to arise.

Let's return to our motivating example: a doctor using algorithmic input to make predictions about patients. These three examples illustrate different possible overall loss rates from a doctor and algorithm with the same average loss rates alone. It is clearly important for overall performance (including patient well-being) to determine whether the human-algorithm system will result in something like Scenario 1 (with higher loss than either the doctor or algorithm alone) or Scenario 3 (with lower loss than either alone). In the coming sections we lay out a theoretical framework for analyzing combined human-algorithm systems more generally to understand their implications for fairness and complementary performance. Our analysis formalizes the observations made in the scenarios above by precisely characterizing the role that ``appropriate'' reliance on the algorithm and variability in performance across regimes plays in complementarity and fairness.

\subsection{Complementarity and fairness}
Complementarity is the best-case scenario for human-algorithm collaboration, and much of our paper will revolve around proving when it can and cannot exist. For instance, in Section  \ref{sec:complementarity} we will give theoretical results describing the kinds of factors necessary for complementarity to be achievable. In particular, we will give conditions on the distributions of loss rates as well as the way predictions are combined. We will show that, all else being equal, complementarity is easier to achieve when distributions of loss for the algorithm and unaided human are highly variable. For example, the algorithm's loss rates are less variable in Tables \ref{tab:nocomp} and \ref{tab:onebetter} (ranging from 0.35 to 0.65) than they are in Table \ref{tab:comp} (where they go from 0.2 to 0.8). However, variable loss rates naturally have fairness implications. 

Fairness concerns are especially salient if the regimes are correlated with sensitive attributes, such as race, ethnicity, sex, gender, or socioeconomic status. One fairness question could revolve around the \enquote{loss disparity} - the difference in loss rates between multiple regimes. Many papers in algorithmic fairness focus on loss disparity (also called accuracy parity or disparate mistreatment \cite{Zafardisparatemistreatment, Fraenkel}), making them a natural focus in our work. In Table \ref{tab:comp}, there's a loss disparity (difference in loss rates between regimes) of 0.8 for the unaided human and 0.6 for the algorithm. However, the combined system has an loss disparity of 0.02---much lower than either the human or algorithm! Can we guarantee that combined systems will always have lower loss disparities? In Section \ref{sec:fair} we will show that, under some conditions, complementarity implies a bounded loss disparity. Another fairness concern could revolve around whether the benefits of incorporating an algorithm are shared among all groups. For example, in Table \ref{tab:comp}, regime 2 sees a reduction in loss when the algorithm is incorporated, going from 1.15 (unaided human) to 0.44 (combined human with algorithm). However, regime 1 sees an increase in loss, from 0.35 to 0.46. Ideally, a combined system should benefit all regimes---but (when) is this possible? In Section \ref{sec:fair}, we will show that, unfortunately, any system exhibiting complementarity can't be one where all regimes see their loss decrease from what it was with the unaided human. 

\section{Complementarity}\label{sec:complementarity}
In this section, we analyze complementarity: when will a combined human-algorithm system have lower average loss than either the unaided human or algorithm? First, we give general results for when complementarity is impossible to achieve. Secondly, we build on these previous results in order to give constructive examples where complementarity is possible. Finally, we discuss some implications of our findings.  All proofs are given in Appendix \ref{app:proof}. 

\subsection{Cases Where Complementarity is Impossible}

This section gives cases where complementarity is \emph{impossible} to achieve. These results help to narrow the scope of cases that we must consider for future analysis. They could also be helpful for practitioners: if any of their cases is addressed by these lemmas, then they can immediately know that their system can never achieve complementarity. 

We begin by presenting two lemmas that concern the distribution of loss rates for the unaided human or the algorithm. Lemma \ref{lem:constant} considers a case where the loss of the algorithm and unaided human are constant across regimes. Constant loss is of course unlikely to arise in practice, but is often a setting considered in the computational science literature. 
\begin{restatable}{lemma}{constant}
\label{lem:constant}
A human-algorithm system where unaided human and algorithm loss rates are constant over regimes can never achieve complementary performance. 
\end{restatable}

The next result, Lemma \ref{lem:cantdom}, considers the setting where one of the components outperforms the other in every regime. Specifically, it says that if one of the unaided human or algorithm always has lower loss than the other, then complementarity is impossible. 
\begin{restatable}{lemma}{cantdom}
\label{lem:cantdom}
Complementarity is impossible if one of the human or algorithm always weakly dominates the loss of the other: that is, if $\alg_i \leq \human_i$ for all $i$, or $\alg_i \geq \human_i$ for all $i$. 
\end{restatable}
\noindent This result may have implications for tasks where the algorithm has extremely high performance, achieving lower loss than the human for all types of inputs. While this could mean that the combined system will have lower loss than the human alone, Lemma \ref{lem:cantdom} tells us that it can't achieve lower loss than both the human and algorithm.  

Next, we will consider two lemmas that concern properties of the combining rules (the way the human incorporates advice from the unaided human and the algorithm). First, Lemma \ref{lem:convex} analyzes a case where the combining function is convex in its arguments.

\begin{restatable}{lemma}{convex}
\label{lem:convex} 
A combining function $\comb(\alg_i, \human_i)$ that is convex in $\alg_i, \human_i$ can never achieve complementary performance. 
\end{restatable}
\noindent Recalling that the maximum function, $\max(\alg, \human)$, that returns whichever of the unaided human or algorithm has higher loss, is convex, this result is very intuitive. 

A simple but important corollary (Corollary \ref{cor:novary}) is that complementarity is impossible whenever the weighting function $\select(\alg_i, \human_i)$ is constant (independent of the algorithm or human's loss rate). This might reflect a situation where the decision-maker either is ignorant of the loss rates by the human or algorithm, or else decides to ignore them. 

\begin{restatable}{corollary}{novary}
\label{cor:novary}
A combining function with a constant weighting function $\select(\alg_i, \human_i) = \select$ can never achieve complementarity performance. 
\end{restatable}
\begin{proof}
Note that $\comb(\alg, \human) = \select \cd \human + (1-\select) \cd \alg$ is convex in both $\alg$ and $\human$. 
\end{proof}

These results, taken together, show that any system that could potentially achieve complementarity must have weighting function that varies with the inputs, must have human or algorithmic loss that varies, must have a combining function that is not convex, and cannot have either the human or algorithm dominate the loss of the other.

\subsection{Cases Where Complementarity is Possible: $\nspace = 2$ regimes}\label{sec:comppossn2}

Having shown cases where complementarity is impossible, in this section we give conditions where complementarity is possible. As we described in Section \ref{sec:modelassump}, our notation describes the loss of the unassisted human in regime $i$ by $\human_i = \Human + \delta_{\human i}$ and the loss of the algorithm by $\alg_i = \Alg + \delta_{\alg i}$, where $\Human$ and $\Alg$ are the average losses of the unassisted human and algorithm, respectively. Additionally, $\select(\alg_i, \human_i)$ (as defined in Equation \ref{eq:weighting}) is the weighting function---the weight that the human places on themselves in making a final prediction.

We will first build intuition with the $\nspace=2$ case: there are exactly two regimes, with regime 1 having probability $\prob$ of occurring. By assumption, $\sum_{i=1}^{\nspace} \prob_i \cd \delta_{\alg i} = 0$ (and similarly for unaided human loss). This allows us to simplify the $\nspace =2$ loss distributions into:

$$\text{unaided human loss: } \begin{cases}
\human_1 = \Human + \delta_h\\
\human_2 = \Human - \frac{p}{1-p}\delta_h
\end{cases} \quad 
\text{algorithmic loss: } 
\begin{cases}
\alg_1 = \Alg + \delta_a\\
\alg_2 = \Alg-\frac{p}{1-p}\delta_a
\end{cases} $$

This formulation is without loss of generality because we allow $\delta_a, \delta_h$ to be positive, negative, or zero. In the case that $\delta_a, \delta_h$ are the same sign, loss rates for the human and algorithm are correlated. If they are of a different sign, then loss rates are anti-correlated. $\delta_a, \delta_h$ values of larger magnitude correspond to a more variable distribution of losses.

Lemma \ref{lem:smallNgeneralbound} gives a condition for when a $\nspace =2$ example can achieve complementarity. Note that the condition gives a lower bound on the magnitude of $\vert \delta_a - \delta_h \vert$, terms which reflect the variability in human and algorithmic loss. This bound depends on $\Alg$ and $\Human$, the average loss rates of the algorithm and human. It also depends on $\select(\alg_1, \human_1), \select(\alg_2, \human_2)$, the weighting functions for the human in the 1st and 2nd regime, respectively. Recall that Corollary \ref{cor:novary} tells us that any system exhibiting complementarity cannot have $\select(\alg_1, \human_1) = \select(\alg_2, \human_2)$, because that would imply a constant weighting function.

\begin{restatable}{lemma}{smallNgeneralbound}
\label{lem:smallNgeneralbound}
Consider the case where $\nspace = 2$, and WLOG assume that $\Alg \leq \Human$: the algorithm has lower average loss than the human. Then, the combined system exhibits complementarity whenever: 
$$(\Human-\Alg) \cd \frac{\select(\alg_1, \human_1) + \frac{1-\prob}{\prob} \cd \select(\alg_2, \human_2)}{\vert \select(\alg_2, \human_2) - \select(\alg_1, \human_1)\vert} < \vert \delta_a -\delta_h\vert$$
\end{restatable}

\begin{figure}
    \centering
    \includegraphics[width=3.5in]{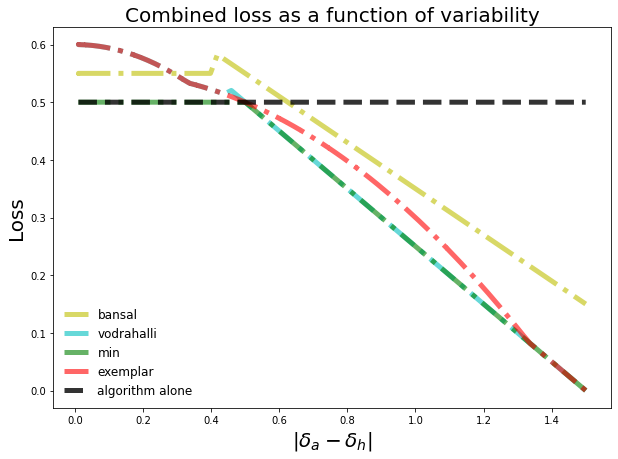}
    \caption{Low combined loss occurs for high variability (large $\vert \delta_a- \delta_h\vert $). The plot displays combined loss for a $N=2$ system, for four combining functions. The black dashed line gives loss of algorithm alone: since for this setting the algorithm has lower error than the unaided human, complementarity occurs below this line.  Note that for all four combining functions, complementarity occurs where $\vert \delta_a - \delta_h\vert$ is large. For details on the combining functions, including the exemplar function (original to this work), see Appendix \ref{app:modelingrules}.}
    \label{fig:comb_err}
    \vspace{-0.3cm}
\end{figure}

Figure \ref{fig:comb_err} displays the loss of the combined system for four different possible combining functions (models of how humans incorporate algorithmic advice). For each of them, complementarity occurs when $\vert \delta_a - \delta_h \vert$ is high, as Lemma \ref{lem:smallNgeneralbound} suggests. However, as Lemma \ref{lem:correrr} states,  this does \emph{not} mean that unaided human and algorithm loss need to be anti-correlated.

\begin{table}[]
\centering 
\begin{tabular}{|c|c|c|c|c|c|}
\hline
                        & \textbf{(Unaided) human} & \textbf{Algorithm} & \textbf{\begin{tabular}[c]{@{}c@{}}Combined\\ (human using algorithm)\end{tabular}}  &    \textbf{\begin{tabular}[c]{@{}c@{}}Weight\\ (on unaided human)\end{tabular}}\\ \hline
\textbf{Regime 1} & 1.15          & 0.2               & 0.44          & 0.25                              \\ \hline
\textbf{Regime 2} & 0.35            & 0.8                & 0.46            &0.75                              \\ \hline 
\textbf{Average}  & 0.75              & 0.5                & 0.45                  &0.5                        \\ \hline
\end{tabular}
\caption{(Reproduced version of Table \ref{tab:comp}). The unaided human has loss 0.75, the algorithm has loss 0.5, and the combined system has loss 0.45 (complementary performance). In this table, we have $\delta_a = 0.4, \delta_h = -0.3$, for $\vert \delta_a - \delta_h \vert = 0.7$. Note that here, losses are anti-correlated: the unaided human has higher loss in regime 1, while the algorithm has higher loss in regime 2. }
\label{tab:compcopied}

\centering 
\begin{tabular}{|c|c|c|c|c|c|}
\hline
                        & \textbf{(Unaided) human} & \textbf{Algorithm} & \textbf{\begin{tabular}[c]{@{}c@{}}Combined\\ (human using algorithm)\end{tabular}}  &    \textbf{\begin{tabular}[c]{@{}c@{}}Weight\\ (on unaided human)\end{tabular}}\\ \hline
\textbf{Regime 1} & 1.48          &   0.53             & 0.77          & 0.25                              \\ \hline
\textbf{Regime 2} & 0.02            & 0.47                & 0.13            &0.75                              \\ \hline 
\textbf{Average}  & 0.75              & 0.5                & 0.45                  &0.5                        \\ \hline
\end{tabular}
\caption{Correlated loss: The combined system (human using algorithm) displays complementary performance even though the unaided human and the algorithm both have higher loss for regime 1. In this table, we have $\delta_a = 0.73, \delta_h = 0.03$, for $\vert \delta_a - \delta_h \vert = 0.7$.}
\label{tab:corr}
\vspace{-0.75cm}
\end{table}

\begin{lemma}\label{lem:correrr}
A system can exhibit complementarity even if the unaided human and the algorithm have correlated loss (both have higher loss in the same regime). 
\end{lemma}

This is a reassuring result because, otherwise, complementarity would not be possible in settings where a given regime is fundamentally harder than another. For example, one regime might be low-resolution images while another might be high-resolution images, so both humans and algorithms are expected to perform worse in the low-resolution regime.  In our notation, if $\delta_a>0$ and $\delta_h > 0$, then both the unaided human and algorithm have lower loss in regime 2. If $\vert \delta_a - \delta_h \vert$ is large, complementarity still may be possible. However, because $\delta_a, \delta_h$ have the same signs, the variability must be larger in order to make $\vert \delta_a - \delta_h \vert$ satisfy the lower bound. 

Tables \ref{tab:compcopied} and \ref{tab:corr} prove Lemma \ref{lem:correrr} through an example. Table \ref{tab:compcopied} is a copy of Scenario 3 (Table \ref{tab:comp}) from the motivating example in Section \ref{sec:motivate}. The values in the table correspond to $\delta_a = 0.4, \delta_h = -0.3$ (giving $\vert \delta_a - \delta_h \vert = 0.7$). Here, the losses are anti-correlated: the unaided human has higher loss in regime 1, while the algorithm has higher loss in regime 2. 

Table \ref{tab:corr} gives an example where the losses are correlated: both the unaided human and algorithm have higher loss for regime 1. The values in the table are given by $\delta_a = 0.03, \delta_h = 0.73$, which again gives $\vert \delta_a - \delta_h \vert = 0.7$. Even though losses are correlated, the overall system still displays complementarity because $\vert \delta_a - \delta_h \vert$ has remained the same, as Lemma \ref{lem:smallNgeneralbound} suggests.

\subsection{Cases Where Complementarity is Possible: $\nspace > 2$ regimes}\label{sec:comppossnbig}

Finally, we will consider the general $\nspace>2$ case. 

\begin{restatable}{lemma}{compN}
\label{lem:compN} 
WLOG, assume that $\Alg \leq \Human$: the algorithm has lower loss, on average. Then, the condition below gives necessary and sufficient conditions for complementarity of the human-algorithm system:
$$(\Human - \Alg) \cd \sum_{i=1}^{\nspace}\prob_i \cd \select(\alg_i, \human_i) < \sum_{i=1}^{\nspace} \prob_i\cd \select(\alg_i, \human_i)  \cd (\delta_{\alg i} - \delta_{\human i})$$
If we view $\select(\alg_i, \human_i)$ and $\delta_{\alg i}, \delta_{\human i}$ as random variables over the instance space with probability mass governed the distribution of instances given by $\{\prob_i\}$, then we can interpret the condition as: 
$$(\Human - \Alg) \cd \mathbb{E}[\select(\alg_i, \human_i)] < Cov\p{\select(\alg_i, \human_i), \delta_{\alg i} - \delta_{\human i}}$$
where $Cov(\cd)$ gives the covariance. 
\end{restatable}

Lemma \ref{lem:compN} gives conditions on complementarity, requiring that the weighting function $\select(\alg_i, \human_i)$ have high covariance with the difference between $\delta_{\alg i}$ and $\delta_{\human i}$. Intuitively, this means that when the algorithm is more \emph{above} its typical loss than the unaided human (when $\delta_{\alg i} > \delta_{\human i}$), then the combined loss should rely more heavily on the unaided human ($\select(\alg_i, \human_i)$ should be large). Conversely, if the algorithm is more \emph{below} its typical loss than the unaided human (when $\delta_{\alg i} < \delta_{\human i}$), then the combined loss should rely more heavily on the algorithm ($\select(\alg_i, \human_i)$ should be small). The lefthand side of the equation lower bounds how large this covariance must, as a function of the gap between the average loss of the unaided human and the algorithm, and the expected value of the weighting function. 

The results in this section imply that unaided human and algorithmic error rates must be highly variable. One question this raises is about achievability: is it even possible to arrange human and algorithmic loss rates like this? For example, achieving a highly variable loss rate for an algorithm may require retraining, and for certain cases may not be possible. Manipulating the loss rate for a human may be possible through re-assigning human effort: for example, assigning multiple humans to certain portions of the input space in order to reduce loss. 

Even if highly variable loss rates are possible, they may not be desirable for other reasons. In particular, the next section discusses the fairness implications of complementarity. 

\section{Fairness}\label{sec:fair}
There are numerous possible notions of fairness that could be relevant for a human-algorithm classifier. In this section, we will analyze three of them and describe how they relate to complementarity. (All proofs are given in Appendix \ref{app:proof}). In general, the notions of fairness concern disparities in loss across different regimes. Disparities in loss rates could be alarming (if they align or correlate with sensitive attributes, such as race, gender, or socioeconomic status) or innocuous (if they are unrelated to sensitive attributes). In this paper, we will consider two general classes of fairness concerns: \emph{fairness of benefit} and \emph{loss disparity} rates. Fairness of benefit relates to which regimes see their loss rate decrease when the algorithm is incorporated into the decision-making process. Ideally, all regimes would benefit from human-algorithm collaboration.  We will show that is unfortunately not possible. Loss disparity rates relate to the gap in loss rates between different regimes: ideally, this gap would be small, which would reflect equal loss rates for all regimes. In this work, we will show that complementarity can sometimes help bound loss disparity.

\subsection{Fairness of benefit}

First, Definition \ref{def:fairben} says that a system has \enquote{fairness of benefit} if all regimes experience a lower loss from the combined human-algorithm system than they would experience with the unassisted human. 
\begin{definition}[Fairness of benefit]\label{def:fairben}
A human/algorithm system exhibits \emph{fairness of benefit} if all regimes benefit from the combined system (experience a lower loss than with the unaided human). 
\end{definition}
\noindent This notion captures the desideratum that the benefits resulting from switching to a combined human-algorithm system are shared by all. However, Lemma \ref{lem:fairbennotcomp} shows that this notion of fairness is incompatible with complementarity. 

\begin{lemma}\label{lem:fairbennotcomp}
Any system exhibiting fairness of benefits cannot have complementary performance. 
\end{lemma}
\begin{proof}
This result is largely due to Lemma \ref{lem:cantdom}. Note that we have: 
$$\comb(\alg_i, \human_i) = (1-\select(\alg_i, \human_i)) \cd \alg_i + \select(\alg_i, \human_i) \cd \human_i$$ 
for $s_i \in [0, 1]$. Then, if we have fairness of benefits, we know:
$$(1-\select(\alg_i, \human_i)) \cd \alg_i + \select(\alg_i, \human_i) \cd \human_i < \human_i \ \forall i \in [N]$$
In order to achieve this, we need $\select(\alg_i, \human_i) <1$ and $\alg_i < \human_i$, for all $i \in [\nspace]$. However, this is exactly the condition addressed in Lemma \ref{lem:cantdom}: the algorithm always has lower loss than the unassisted human, which means that complementarity is impossible. 
\end{proof}

Note that Definition \ref{def:fairben} could be defined symmetrically, defining \enquote{benefit} as a reduction in loss as compared to the algorithmic prediction. In this case, an analogous version of Lemma \ref{lem:fairbennotcomp} would hold, similarly showing that this notion of fairness is incompatible with complementarity.

Lemma \ref{lem:fairbennotcomp} tells us that there is an inherent tension between achieving complementarity and ensuring all people who use the system see their loss rates decrease. For certain application areas, fairness of benefit might be more important than complementarity, so practitioners might consciously choose to prioritize it. On the other hand, a practitioner who opts to achieve complementarity instead of fairness of benefit might wish to consider alternate ways to support any groups that see their loss increase in the combined system.

\subsection{Loss disparity}

Next, Definition \ref{def:errdisp} describes fairness as the disparity in loss rates between different regimes. Intuitively, this notion relates to group-based notions of fairness: loss rates should be relatively similar between members of different groups. However, Lemma \ref{lem:fairinputbound} again describes how this notion of fairness is in tension with complementarity, which puts a lower bound on the level of this kind of unfairness. 

\begin{definition}[Loss disparity]\label{def:errdisp}
A prediction system exhibits $\epsilon$-loss disparity if the losses in different regimes differ by no more than $\epsilon$. We will use $\epsilon_h, \epsilon_a, \epsilon_c$ to refer to the loss disparity of the unaided human, algorithm alone, and combined human-algorithm system, respectively. 
\end{definition}

\begin{restatable}{lemma}{fairinputbound}
\label{lem:fairinputbound}
WLOG assume that $\Alg \leq \Human$: the algorithm has lower average loss than the human. Then, any system exhibiting complementarity has a lower bound on $\epsilon_a+ \epsilon_h$: the combined loss disparity of the unaided human and algorithm.
$$\Alg-C + (\Human-\Alg) \cd \sum_{i=1}^{\nspace} \prob_i \cd \select(\alg_i, \human_i) < \epsilon_\alg + \epsilon_h$$
\end{restatable}
Lemma \ref{lem:fairinputbound} should match our intuition from  Section \ref{sec:complementarity}. There, results indicated that complementarity is easiest to achieve when loss rates are highly variable. However, this directly contradicts with the goal of minimizing loss disparity. 

Next, we will consider the loss disparity of the combined human-algorithm system. Unfortunately, it is possible for a system exhibiting complementarity to \emph{exacerbate} unfairness. Table \ref{tab:n3counterex} gives an example where the combined system has a higher loss disparity than either the unaided human or the algorithm, even though it exhibits complementarity. This result (which parallels results from \cite{dwork2018fairness, dwork2020individual, wang2019practical}) means that practitioners must be careful when discussing fairness of combined systems: fairness guarantees from the individual components don't necessarily transfer to the combined human-algorithm system. However, Lemma \ref{lem:fairupbound} gives a condition where the loss disparity of the combined human-algorithm system is upper bounded by the loss disparity of the unaided human or the algorithm. 

\begin{table}[]
\centering 
\begin{tabular}{|c|c|c|c|}
\hline
\textbf{}               & \textbf{(Unaided) Human} & \textbf{Algorithm} & \textbf{Combined (human using algorithm)} \\ \hline
\textbf{Regime 1} & 0.95           & 0.85               & 0.895                   \\ \hline
\textbf{Regime 2} & 0.95           & 0.02               & 0.05                    \\ \hline
\textbf{Regime 3} & 0.15           & 0.45               & 0.255                   \\ \hline
\textbf{Average loss}  & 0.68           & 0.44               & 0.40                    \\ \hline
\end{tabular}
\caption{This system exhibits complementarity, since average loss is lowest in the combined human-algorithm system. However, note that loss disparity is \emph{increased} in the combined system: $\epsilon_h = 0.8, \epsilon_a = 0.83$, but $\epsilon_c = 0.84$. }
\label{tab:n3counterex}
\vspace{-0.8cm}
\end{table}

\begin{restatable}{lemma}{fairupbound}
\label{lem:fairupbound}
Define $i+$ as the regime where the combined human-algorithm system has highest loss and $i-$ as the regime where it has lowest loss. Then, the loss disparity of the combined system is upper bounded by the loss disparity of the unaided human or algorithm, so long as neither the unaided human or algorithm dominates the other in both $i+, i-$. That is, 
$$\text{If either case is satisfied: }\begin{cases}
\human_{i+} \leq\alg_{i+} \text{ and } \human_{i-} \geq\alg_{i-}\\
\human_{i+} \geq\alg_{i+} \text{ and } \human_{i-} \leq\alg_{i-}\\
\end{cases} \quad \Rightarrow \quad \epsilon_c \leq \max(\epsilon_a, \epsilon_h)$$
\end{restatable}

This last lemma gives our first positive result for fairness: it gives conditions where human-algorithm system exhibiting complementarity at least doesn't exacerbate unfairness. Specifically, what it requires is that neither the unaided human nor the algorithm dominates the other for both of the most extreme regimes (where the combined system has highest and lowest loss). As we would expect, the scenario in Table \ref{tab:n3counterex} violates this: the unaided human dominates in both regime 1 (highest loss) and regime 2 (lowest loss), which allows the combined loss disparity $\epsilon_c$ to be greater than $\max(\epsilon_a, \epsilon_h)$.  Interestingly, Lemma \ref{lem:fairupbound} is quite powerful: it only relies on the losses within two specific regimes and holds regardless of whether the overall system satisfies complementarity. Practitioners could use Lemma \ref{lem:fairupbound} to guide their algorithm development: so long as the preconditions are satisfied, they can guarantee that the combined system will never exacerbate unfairness. 

As we mentioned previously, these definitions are only a few of the possible fairness concerns we could analyze. However, this analysis highlights the importance of considering fairness, especially ways that it might be in tension with achieving complementarity.

\section{Conclusion and Future Directions}

In this work, we introduce a simple theoretical model of human-algorithm collaboration, which we show is flexible enough to encompass models analyzed in prior work. 
Using this model, we obtain theoretical impossibility results that characterize settings where complementarity is not achievable.
We also use this framework to construct cases where complementarity is possible, given certain conditions on the loss distributions. Finally, we consider the implications, especially fairness, of the requirements in order to achieve complementarity. 

Our approach admits multiple possible avenues for future work. Our work highlights the importance of variable loss rates: algorithmic loss that is not constant over the input space. However, as mentioned previously, it may not be possible to achieve extremely variable loss rates. Future work could model algorithmic loss rates more explicitly, describing loss distributions that are both achievable and lead to complementary performance. For example, some prior work has demonstrated that, for many algorithms, reducing loss becomes harder as the level of loss decreases, which could make it more difficult to achieve highly variable loss rates \citep{hestness2017deep, kaplan2020scaling}. 

Similarly, future work could relax assumptions made in our work. Relaxing Assumption \ref{assump:bounded}, for example, could involve analyzing combining rules that, on any individual regime, could do better or worse than the human or algorithmic input loss rates. Relaxing Assumption \ref{assump:lossonly} could involve modeling cases where regimes with identical human and algorithmic loss rates might be treated differently by the combiner. Any of these future analyses could allow us to have greater insight into a variety of ways human-algorithm systems perform. 

\ifarxiv
\else 
\begin{acks} 
\fi 

\ifarxiv 
\section*{Acknowledgments}
\else 
\fi 
This work was supported in part by NSF grant DGE-1650441. We are grateful to Vijay Keswani, Jon Kleinberg, Pang Wei Koh, Michela Meister, Emma Pierson, Charvi Rastogi, Aaron Roth, Kiran Tomlinson, Aaron Tucker, Qian Yang, Joyce Zhou, James Zou, the AWS Machine Learning team, the AI, Policy, and Practice working group at Cornell, and the attendees at the NeurIPS 2021 Workshops on Human-Centered AI and Human and Machine Decisions for invaluable discussions.

\ifarxiv 
\else 
\end{acks}
\fi 

\bibliographystyle{ACM-Reference-Format}
\bibliography{sample-base}

\appendix

\section{Modeling Combining Functions}\label{app:modelingrules}

The definitions presented in Section \ref{sec:modelassump} reflect multiple ways humans could incorporate algorithmic inputs. However, their functional forms in some cases are less clean to analyze. Definition \ref{def:gap} gives an \emph{exemplar} weighting function (original to this work) that we created in order to illustrate common patterns in weighting functions, while also allowing for tractable theoretical analysis. In this function, for $m>0$, the combining rule is more likely to select the algorithm when $\alg < \human$ (when the algorithm has lower loss rate). For $m>0$, the reverse is true: the combining rule \enquote{mistakenly} goes with the input with higher loss.

\begin{definition}[Exemplar weighting function]\label{def:gap}
The exemplar weighting function is given by: 
$$\select(\alg, \human) =\begin{cases}
b- m \cd (\human-\alg) & 0 \leq b- m \cd (\human-\alg) \leq 0\\
0 & b- m \cd (\human-\alg) < 0\\
1 & b- m \cd (\human-\alg) > 1
\end{cases}$$
\end{definition}
 
Figure \ref{fig:prob_select} plots the weighting function for each of the selection rules presented. Note that, in general, as the human gets higher loss, the weight on the human decreases. 

\begin{figure}
    \centering
    \includegraphics[width=4in]{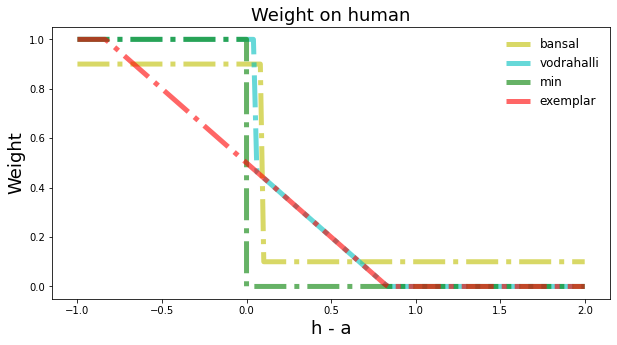}
    \caption{The weight on the human, given a difference $\human-\alg$ between algorithmic and human loss, for multiple weighting functions. Note that, in general, as the human gets higher loss than the algorithm, the weight on the human decreases.}
    \label{fig:prob_select}
\end{figure}

Lemma \ref{lem:linearcompsym}, below, gives an analgous version of Lemma \ref{lem:smallNgeneralbound} for the exemplar function. Note that it similarly shows that complementarity occurs when $\vert \delta_a - \delta_h\vert$ is large (when losses are highly variable). 

\begin{restatable}{lemma}{linearcompsym}
\label{lem:linearcompsym}
Consider the exemplar weighting function with $N =2$ and $\select(\alg_0, \human_0), \select(\alg_1, \human_1) < 1$, and where $\Alg \leq \Human$. Then, the system exhibits complementarity whenever: 
$$\sqrt{\Human - \Alg} \cd \sqrt{\frac{1-p}{p} \cd \p{\frac{1-b}{m} - (\Human-\Alg)}} < \vert \delta_a - \delta_h \vert $$
\end{restatable}

\section{Proofs}\label{app:proof}

\equivalent*

\begin{proof}
We will define a weighting function as follows: 
$$\select(\alg, \human) = \begin{cases}
\frac{\comb(\alg, \human) - \alg}{\human - \alg} & \alg \ne \human \\
\frac{1}{2} & \alg = \human
\end{cases}$$
Note that if $\comb(\alg, \human) = \human$, then $\select(\alg, \human) = 1$ (the weighting function puts all weight on the unaided human), and if $\comb(\alg, \human)= \alg$, then $\select(\alg, \human) = 0$ (the weighting function puts no weight on the unaided human).

First, we'll show that $\select(\alg, \human) \in [0, 1]$. If $\alg = \human$, this is true by construction (because $\select(\alg, \human) = \frac{1}{2}$). What we want to show is: 
$$0 \leq \frac{\comb(\alg, \human) - \alg}{\human - \alg} \leq 1$$
If $\human > \alg$, then this is equivalent to requiring: 
$$0 \leq \comb(\alg, \human)- \alg \leq \human - \alg$$
The lefthand inequality is satisfied because $\comb(\alg, \human) \geq \min(\alg, \human) = \alg $ (in this case). The righthand inequality is satisfied because $\comb(\alg, \human) \leq \max(\alg, \human) = \human$ (in this case). 

On the other hand, if $\human < \alg $, then the inequality we want to show is: 
$$0 \geq \comb(\alg, \human) - \alg \geq \human - \alg$$
Again, the lefthand inequality is satisfied because $\comb(\alg, \human) \leq \max(\alg, \human) = \alg $ (in this case). The righthand inequality is satisfied because $\comb(\alg, \human) \geq \min(\alg, \human) = \human $ (in this case). 

Next, we'll show that this weighting function is correct: that is, that: 
$$\comb(\alg, \human) = (1-\select(\alg, \human)) \cd \alg + \select(\alg, \human) \cd \human$$
Note that Assumption \ref{assump:bounded} means that if $\alg = \human$, then $\comb(\alg, \human) = \alg = \human$, so any $\select(\alg, \human) \in [0, 1]$ would result in a correct weighting function.  For $\alg \ne \human$, we can write: 
\begin{align*}
& (1-\select(\alg, \human)) \cd \alg + \select(\alg, \human) \cd \human \\
 = &  \alg \cd \p{1-\frac{\comb(\alg, \human) - \alg}{\human - \alg}} + \human \cd \frac{\comb(\alg, \human) - \alg}{\human - \alg} \\
 = & \alg \cd \frac{\human - \alg - \comb(\alg, \human) + \alg}{\human -\alg} + \human \cd \frac{\comb(\alg, \human) - \alg}{\human - \alg}\\
 = & \alg \cd \frac{\human - \comb(\alg, \human)}{\human -\alg} + \human \cd \frac{\comb(\alg, \human) - \alg}{\human - \alg}\\
 = & \frac{\alg \cd \human - \alg \cd \comb(\alg, \human) + \human \cd \comb(\alg, \human) - \alg \cd \human}{\human - \alg}\\
 = &\frac{\comb(\alg, \human) \cd (\human - \alg)}{\human - \alg}\\
 = & \comb(\alg, \human)
\end{align*}
as desired. 
\end{proof}

\constant*

\begin{proof}
Constant loss rates means that $\alg_i = \Alg$ and $\human_i = \Human$ for all $i \in [\nspace]$. The combined system has loss: 
\begin{equation*}
\sum_{i \in [\nspace]}\prob_i \cd \comb(\alg_i, \human_i) = \sum_{i \in [\nspace]}\prob_i \cd \comb(\Alg, \Human) = \comb(\Alg, \Human) \geq \min(\Alg,\Human) 
\end{equation*}
\end{proof}

\cantdom*

\begin{proof}
WLOG, we will assume that $\Alg \leq \Human$. Then, complementarity occurs when: 
\begin{align*}
\sum_{i =1}^{\nspace} \prob_i \cd \comb(\alg_i, \human_i) < & \Alg\\
\sum_{i =1}^{\nspace}\prob_i \cd ((1-\select(\alg_i, \human_i)) \cd \alg_i + \select(\alg_i, \human_i) \cd \human_i) < & \sum_{i=1}^{\nspace} \prob_i \cd \alg_i\\
\sum_{i =1}^{\nspace} \prob_i \cd (-\select(\alg_i, \human_i) \alg_i + \select(\alg_i, \human_i) \cd \human_i) <  & 0\\
\sum_{i =1}^{\nspace} \prob_i \cd \select(\alg_i, \human_i)\cd (\human_i - \alg_i) < & 0
\end{align*}

If the algorithm weakly dominates the unaided human ($\human_i \geq \alg_i$ for all $i$), then this inequality can never be satisfied because the entire lefthand side is positive or 0. If the unaided human weakly dominates the algorithm ($\alg_i \geq \human_i$ for all $i$) then $\Alg \geq \Human$. From our previous assumption that $\Alg \leq \Human$, which means that we must have that $\Alg = \Human$: both have identical average losses.

We will now show that further analysis means in the case that $\Alg = \Human$, complementarity is still impossible. If $\Alg = \Human$, we could have equivalently written the complementarity condition as:
\begin{align*}
\sum_{i =1}^{\nspace} \prob_i \cd \comb(\alg_i, \human_i) < & \Human = \Alg \\
\sum_{i =1}^{\nspace}\prob_i \cd ((1-\select(\alg_i, \human_i)) \cd \alg_i + \select(\alg_i, \human_i) \cd \human_i) <  &\sum_{i=1}^{\nspace} \prob_i \cd \human_i\\
\sum_{i =1}^{\nspace} \prob_i \cd ((1-\select(\alg_i, \human_i)) \cd \alg_i +(1-\select(\alg_i, \human_i)) \cd \human_i) < & 0\\
\sum_{i =1}^{\nspace} \prob_i \cd (1-\select(\alg_i, \human_i))\cd (\alg_i - \human_i) <& 0
\end{align*}

In this case, we've assumed that $\alg_i \geq \human_i$. However, this means that the lefthand side of the above inequality is positive, which means complementarity isn't satisfied. 
\end{proof}

\convex*

\begin{proof}
The first inequality in this proof is by Jensen's inequality and the last inequality is due to our construction of the combining function:
\begin{equation*}\sum_{i \in [\nspace]}\prob_i \cd \comb(\alg_i, \human_i) \geq \comb\p{\sum_{i \in [\nspace]}\prob_i  \cd \alg_i, \sum_{i \in [\nspace]}\prob_i \cd \human_i} = \comb(\Alg, \Human) \geq \min(\Alg, \Human)
\end{equation*}
\end{proof}

\smallNgeneralbound*

\begin{proof}
If we assume that $\Alg \leq \Human$, then complementarity occurs whenever $C < \Alg$, or:
The system exhibits complementarity when: 
\begin{align*}
p \cd \comb(\alg_1, \human_1) + (1-p) \cd \comb(\alg_2, \human_2)< & \Alg\\
p \cd ((1-\select(\alg_1, \human_1)) \cd \alg_1 + \select(\alg_1, \human_1) \cd \human_1) + (1-p) \cd ((1-\select(\alg_2, \human_2) \cd \alg_2 + \select(\alg_2, \human_2) \cd \human_2) < & p \cd \alg_1 + (1-p) \cd \alg_2\\
\prob \cd \select(\alg_1, \human_1) \cd (\human_1 - \alg_1) + (1-\prob) \cd \select(\alg_2, \human_2) \cd (\human_2 - \alg_2) < & 0
\end{align*}

We insert the values of $\alg_i, \human_i$, to get: 
\begin{align*}
\prob \cd \select(\alg_1, \human_1) \cd (\Human + \delta_h - \Alg - \delta_a) + (1-\prob) \cd \select(\alg_2, \human_2) \cd \p{\Human - \frac{\prob}{1-\prob} \cd \delta_h - \Alg + \frac{\prob}{1-\prob} \cd \delta_a} < &  0\\
(\Human-\Alg) \cd (\prob \cd \select(\alg_1, \human_1) + (1-\prob) \cd \select(\alg_2, \human_2)) +\prob \cd \select(\alg_1, \human_1) \cd ( \delta_h -\delta_a) -\prob \cd \select(\alg_2, \human_2) \cd \p{ \delta_h - \delta_a} < & 0\\
(\Human-\Alg) \cd (\prob \cd \select(\alg_1, \human_1) + (1-\prob) \cd \select(\alg_2, \human_2)) -\prob\cd (\delta_a -\delta_h)\cd (\select(\alg_1, \human_1) - \select(\alg_2, \human_2)) < & 0 
\end{align*}

Note that we have assumed $\Human \geq \Alg$, so in order for this inequality to hold, we must have that the other term (involving $\delta_a, \delta_h$) be positive. Note that by Lemma \ref{lem:constant} we cannot have $\select(\alg_1, \human_1) = \select(\alg_2, \human_2)$ if it exhibits complementarity. 
In order to have the inequality satisfied, one of two conditions must hold: 

\begin{itemize}
    \item \textbf{Case 1:} $\delta_a > \delta_h$ and $\select(\alg_1, \human_1) >\select(\alg_2, \human_2)$ (the unaided human is weighted more heavily in regime 1). Note that in this case, we must have $\human_2 >\alg_2$, because: 
    
    \begin{align*}
    \human_2 > &  \alg_2\\
    \Human - \frac{\prob}{1-\prob} \cd \delta_h >  &\Alg - \frac{\prob}{1-\prob} \cd \delta_a\\
    \Human - \Alg - \frac{\prob}{1-\prob} \cd (\delta_h - \delta_a) >& 0\\
    \Human - \Alg + \frac{\prob}{1-\prob} \cd (\delta_a - \delta_h) > &0
    \end{align*}
    which is satisfied because $\Human - \Alg \geq 0$ and $\delta_a > \delta_h$. By Lemma \ref{lem:cantdom}, we know that $\human_2 > \alg_2$ implies that $\human_1 \leq \alg_1$. Taken together with $\select(\alg_1, \human_1) > \select(\alg_2, \human_2)$, this means that the combined system must weight the human more heavily in instance 1, where it has lower loss than the human. 
    \item \textbf{Case 2:} $\delta_a < \delta_h $ and $\select(\alg_1, \human_1) <\select(\alg_2, \human_2)$ (the unaided human is weighted more heavily in regime 2). In this case, we must have $\human_1 > \alg_1$, because: 
    
    \begin{align*}
    \Human + \delta_h > & \Alg + \delta_a \\
    \Human - \Alg + \delta_h - \delta_a >  & 0 
    \end{align*}
    This must be satisfied because $\Human - \Alg>0$ and $\delta_h > \delta_a$. By similar reasoning to above, this means that $\human_2 \leq \alg_2$. Taken with $\select(\alg_1, \human_1) <\select(\alg_2, \human_2)$, this again means that the system must weight the unaided more heavily in the regime where it has lower loss. 
\end{itemize}
Finally, we can simplify the inequality: 
$$(\Human-\Alg) \cd (\prob \cd \select(\alg_1, \human_1) + (1-\prob) \cd \select(\alg_2, \human_2)) < \prob\cd (\delta_a -\delta_h)\cd (\select(\alg_1, \human_1) - \select(\alg_2, \human_2)) $$
If $\delta_a > \delta_h$, then we know that $\select(\alg_1, \human_1) > \select(\alg_2, \human_2)$, so we can rewrite this as: 
$$(\Human - \Alg) \cd \frac{\select(\alg_1, \human_1) + \frac{1-\prob}{\prob} \cd \select(\alg_2, \human_2)}{\select(\alg_1, \human_1) - \select(\alg_2, \human_2)} \leq \delta_a- \delta_h$$
On the other hand, if $\delta_a < \delta_h $, we know that $\select(\alg_1, \human_1) < \select(\alg_2, \human_2)$, so we can rewrite this as: 
\begin{align*}
(\Human-\Alg) \cd (\prob \cd \select(\alg_1, \human_1) + (1-\prob) \cd \select(\alg_2, \human_2)) < & \prob\cd (\delta_h -\delta_a)\cd (\select(\alg_2, \human_2) - \select(\alg_1, \human_1))\\
(\Human-\Alg) \cd \frac{\select(\alg_1, \human_1) + \frac{1-\prob}{\prob} \cd \select(\alg_2, \human_2)}{(\select(\alg_2, \human_2) - \select(\alg_1, \human_1)} < & \delta_a- \delta_h)
\end{align*}
Either way simplifies to: 
$$(\Human-\Alg) \cd \frac{\select(\alg_1, \human_1) + \frac{1-\prob}{\prob} \cd \select(\alg_2, \human_2)}{\vert \select(\alg_2, \human_2) - \select(\alg_1, \human_1)\vert} <\vert (\delta_a -\delta_h)\vert $$
\end{proof}

\compN*

\begin{proof}
If we assume that $\Alg \leq \Human$, then complementarity occurs whenever $C < \Alg$, or: 
\begin{align*}
\sum_{i=1}^{\nspace} \prob_i \cd \comb(\alg_i, \human_i) \leq & \sum_{i=1}^{\nspace} \prob_i \cd \alg_i\\
\sum_{i=1}^{\nspace} \prob_i \cd ((1-\select(\alg_i, \human_i)) \cd \alg_i + \select(\alg_i, \human_i) \cd \human_i) \leq  &\sum_{i=1}^{\nspace} \prob_i \cd \alg_i\\
\sum_{i=1}^{\nspace} \prob_i \cd (-\select(\alg_i, \human_i) \cd \alg_i + \select(\alg_i, \human_i) \cd \human_i) \leq & 0\\
\sum_{i=1}^{\nspace} \prob_i \cd \select(\alg_i, \human_i) \cd (\human_i - \alg_i) \leq  &0
\end{align*}
Plugging in for the values gives: 
\begin{align*}
\sum_{i=1}^{\nspace} \prob_i \cd \select(\alg_i, \human_i) \cd (\Human + \delta_{\human i} - \Alg - \delta_{\alg i}) \leq & 0\\
(\Human - \Alg) \cd \sum_{i=1}^{\nspace} \prob_i \cd \select(\alg_i, \human_i) < & \sum_{i=1}^{\nspace}\prob_i \cd \select(\alg_i, \human_i) \cd (\delta_{\alg i} - \delta_{\human i})
\end{align*}
as desired. 
\end{proof}

\fairinputbound*
\begin{proof}
First, we will use the complementarity result from partway through the proof of Lemma \ref{lem:compN}. We calculate the difference between the average algorithmic loss and the average combined human-algorithmic loss (which gives us the average benefit of collaboration): 
\begin{align*}
\Alg - C= & \sum_{i=1}^{\nspace} \prob_i \cd \alg_i - \sum_{i=1}^{\nspace} \prob_i \cd \comb(\alg_i, \human_i) \\
 = & \sum_{i=1}^{\nspace} \prob_i \cd \alg_i - \sum_{i=1}^{\nspace} \prob_i \cd ((1-\select(\alg_i, \human_i)) \cd \alg_i + \select(\alg_i, \human_i) \cd \human_i) \\
= &\sum_{i=1}^{\nspace} \prob_i \cd (\alg_i - \alg_i + \select(\alg_i, \human_i) \cd \alg_i - \select(\alg_i, \human_i) \cd \human_i)\\
= & \sum_{i=1}^{\nspace} \prob_i \cd \select(\alg_i, \human_i) \cd (\alg_i - \human_i)\\
= & \sum_{i=1}^{\nspace} \prob_i \cd \select(\alg_i, \human_i) \cd (\Alg + \delta_{\alg i} - \Human - \delta_{\human i})\\
 = & (\Alg - \Human) \cd \sum_{i=1}^{\nspace} \prob_i \cd \select(\alg_i, \human_i) + \sum_{i=1}^{\nspace} \prob_i \cd \select(\alg_i, \human_i) \cd (\delta_{\alg i} - \delta_{\human i})\\
\end{align*}
Rearranging gives: 
$$\Alg - C+ (\Human - \Alg) \cd \sum_{i=1}^{\nspace} \prob_i \cd \select(\alg_i, \human_i) = \sum_{i=1}^{\nspace} \prob_i \cd \select(\alg_i, \human_i) \cd (\delta_{\alg i} - \delta_{\human i})$$

Next, we will define: 
$$\epsilon_a = \alg_{\alg+} - \alg_{\alg-} = \Alg + \delta_{\alg +} - \Alg - \delta_{\alg-} = \delta_{\alg+} - \delta_{\alg-}$$
$$\epsilon_h = \human_{\human+} - \human_{\human-} = \Human + \delta_{\human +} - \Human - \delta_{\human-} = \delta_{\human+} - \delta_{\human-}$$
where $\alg+, \alg-, \human+, \human-$ are the indices of the maximum and minimum loss for the algorithm and unaided human, respectively. Note that we require $\sum_{i=1}^{\nspace} \prob_i \cd \delta_{\alg i} = \sum_{i=1}^{\nspace} \prob_i \cd \delta_{\human i} = 0$, which implies that $\delta_{\alg+}, \delta_{\human+} \geq 0$ and $\delta_{\alg-}, \delta_{\human-} \leq 0$. Then, we know that: 
$$\epsilon_a + \epsilon_h =  \delta_{\alg+} - \delta_{\alg-} +\delta_{\human+} - \delta_{\human-} > \delta_{\alg+} - \delta_{\human-} $$
Define $\mathcal{P} = \{i \ \vert \ \delta_{\alg i}\geq \delta_{\human i}$ and $\mathcal{N} = \{i \ \vert \ \delta_{\alg i}< \delta_{\human i}$. By definition, $\delta_{\alg+} \geq \delta_{\alg i}$ and $\delta_{\human -} \leq \delta_{\human i}$ for all $i \in [\nspace]$. Then, we know that: 
$$\delta_{\alg+} - \delta_{\human -} \geq \sum_{i \in \mathcal{P}}\prob_i \cd (\delta_{\human i} - \delta_{\alg i}) \geq \sum_{i \in \mathcal{P}} \prob_i \cd \select(\alg_i, \human_i) \cd  (\delta_{\human i} - \delta_{\alg i})$$
where we have used the fact that $\select(\alg_i, \human_i) \leq 1$. Finally, we know that: 
$$\sum_{i \in \mathcal{P}} \prob_i \cd \select(\alg_i, \human_i) \cd  (\delta_{\human i} - \delta_{\alg i}) \geq \sum_{i \in \mathcal{P}} \prob_i \cd \select(\alg_i, \human_i) \cd  (\delta_{\human i} - \delta_{\alg i}) + \sum_{i \in \mathcal{N}} \prob_i \cd \select(\alg_i, \human_i) \cd  (\delta_{\human i} - \delta_{\alg i}) = \sum_{i=1}^{\nspace} \prob_i \cd \select(\alg_i, \human_i) \cd  (\delta_{\human i} - \delta_{\alg i})$$
where the first inequality comes by the fact that $\delta_{\human i} - \delta_{\alg i}$ for $i \in \mathcal{N}$. Finally, we can combine this analysis with our previous analysis on the gap in loss rate between the algorithm and the combined system:  

$$\Alg - C+ (\Human - \Alg) \cd \sum_{i=1}^{\nspace} \prob_i \cd \select(\alg_i, \human_i) = \sum_{i=1}^{\nspace} \prob_i \cd \select(\alg_i, \human_i) \cd (\delta_{\alg i} - \delta_{\human i}) \leq \delta_{\alg +} - \delta_{\human -} < \epsilon_a + \epsilon_h$$
as desired. 
\end{proof}

\fairupbound*
\begin{proof}
We wish to upper bound $\epsilon_c$, which is given by: 
$$\epsilon_c = \comb(\alg_{i+}, \human_{i+}) - \comb(\alg_{i-}, \human_{i-})$$
such that: 
$$i+ = \text{argmax}_{i \in [\nspace]} \comb(\alg_i,\human_i) \quad i- = \text{argmin}_{i \in [\nspace]} \comb(\alg_i,\human_i)$$
Note that by Assumption \ref{assump:bounded}, we must have: 
$$\comb(\alg_{i+}, \human_{i+}) \leq \max(\alg_{i+}, \human_{i+})\text{ and } 
\comb(\alg_{i-}, \human_{i-}) \geq \min(\alg_{i-}, \human_{i-})$$
The statement of this lemma gives two cases, which we will consider in turn. \\
In Case 1, we assume that: 
$$\human_{i+} \leq\alg_{i+} \text{ and } \human_{i-} \geq\alg_{i-}$$
In this case, we know that: 
$$\comb(\alg_{i+}, \human_{i+}) - \comb(\alg_{i-}, \human_{i-}) \leq \alg_{i+} - \alg_{i-} \leq \epsilon_a \leq \max(\epsilon_a, \epsilon_h)$$
In Case 2, we assume that: 
$$\human_{i+} \geq\alg_{i+} \text{ and } \human_{i-} \leq\alg_{i-}$$
In this case, we know that: 
$$\comb(\alg_{i+}, \human_{i+}) - \comb(\alg_{i-}, \human_{i-}) \leq \human_{i+} - \human_{i-} \leq \epsilon_h \leq \max(\epsilon_a, \epsilon_h)$$
In either case, $\epsilon_c \leq \max(\epsilon_a, \epsilon_h)$
\end{proof}

\linearcompsym*

\begin{proof}
This lemma is a more specific version of Lemma \ref{lem:smallNgeneralbound}, so we start with an intermediate result from that proof. Complementarity occurs whenever: 
$$\prob \cd \select(\alg_1, \human_1) \cd (\human_1 - \alg_1) + (1-\prob) \cd \select(\alg_2, \human_2) \cd (\human_2 - \alg_2) <  0$$
For the exemplar combining rule, we have that: 
$$\select(\alg_i, \human_i) = b - m \cd (\human_i - \alg_i)$$
Plugging in for the values of $\select(\alg_1, \human_1), \select(\alg_2, \human_2)$ gives: 
\begin{align*}
\prob \cd (b - m \cd (\human_1 - \alg_1)) \cd (\human_1 - \alg_1) + (1-\prob) \cd (b - m \cd (\human_2 - \alg_2))\cd (\human_2 - \alg_2) < &   0\\
b \cd (\prob  \cd (\human_1 - \alg_1) + (1-\prob) \cd (\human_2 - \alg_2)) -m (\prob \cd (\human_1 - \alg_1)^2 + (1-\prob) \cd (\human_2 - \alg_2)^2) < & 0 \\
b \cd (\prob  \cd \human_1 - \prob \cd \alg_1 + (1-\prob) \cd \human_2 - (1-\prob) \cd \alg_2)-m (\prob \cd \cd (\human_1 - \alg_1)^2 + (1-\prob) \cd (\human_2 - \alg_2)^2) < & 0 \\
b \cd (\Human - \Alg)-m (\prob \cd (\human_1 - \alg_1)^2 + (1-\prob) \cd (\human_2 - \alg_2)^2) < & 0 \\
\end{align*}
We can analyze the term with the $m$ coefficient by plugging in for values of $\alg_i, \human_i$. 
\begin{align*}
 & \prob \cd (\human_1 - \alg_1)^2 + (1-\prob) \cd (\human_2 - \alg_2)^2\\
 =& \prob \cd(\Human + \delta_h - \Alg - \delta_a)^2 + (1-\prob) \cd \p{\Human - \frac{\prob}{1-\prob} \cd \delta_h - \Alg + \frac{\prob}{1-\prob} \cd \delta_a}^2 
\end{align*}
We expand out each to get: 
\begin{align*}
= & \prob \cd (\Human - \Alg)^2 + \prob\cd (\delta_h - \delta_a)^2 + 2 \prob \cd (\Human - \Alg) \cd (\delta_h - \delta_a) + (1-\prob) \cd (\Human - \Alg)^2 \\
& +(1-\prob) \cd \frac{\prob^2}{(1-\prob)^2} \cd (\delta_a - \delta_h)^2 + 2 \cd (1-\prob) \cd (\Human - \Alg) \cd \frac{\prob}{1-\prob} \cd (\delta_a - \delta_h)
\end{align*}
We note that two of the terms cancel: 
\begin{align*}
 & 2 \prob \cd (\Human - \Alg) \cd (\delta_h - \delta_a)  +  2 \cd (1-\prob) \cd (\Human - \Alg) \cd \frac{\prob}{1-\prob} \cd (\delta_a - \delta_h)\\
  = &  2 \prob \cd (\Human - \Alg) \cd (\delta_h - \delta_a)  +  2 \cd \prob \cd (\Human - \Alg) \cd (\delta_a - \delta_h)\\
 = & 0
\end{align*}
Next, we can simplify the other terms: 
\begin{align*}
= & \prob \cd (\Human - \Alg)^2 + \prob\cd (\delta_h - \delta_a)^2 + (1-\prob) \cd (\Human - \Alg)^2 + (1-\prob) \cd \frac{\prob^2}{(1-\prob)^2} \cd (\delta_a - \delta_h)^2 \\
= & (\Human - \Alg)^2 + \prob (\delta_h - \delta_a)^2 + \frac{\prob^2}{1-\prob} \cd (\delta_a - \delta_h)^2\\
= & (\Human - \Alg)^2 + (\delta_h - \delta_a)^2 \cd \p{\prob + \frac{\prob^2}{1-\prob}}\\
= & (\Human - \Alg)^2 + (\delta_h - \delta_a)^2 \cd \frac{\prob}{1-\prob}\\
\end{align*}
where we have used that $\prob + \frac{\prob^2}{1-\prob} = \frac{\prob-\prob^2 + \prob^2}{1-\prob} = \frac{\prob}{1-\prob}$. We can combine this with the inequality we were analyzing earlier to get:
\begin{align*}
(1-b) \cd (\Human - \Alg) -m \cd \p{ (\Human - \Alg)^2 + (\delta_h - \delta_a)^2 \cd \frac{\prob}{1-\prob}} < &  0\\
(1-b) \cd (\Human - \Alg) -m \cd(\Human - \Alg)^2 <  & m \cd  (\delta_h - \delta_a)^2 \cd \frac{\prob}{1-\prob} \\
\frac{1-\prob}{\prob} \p{\frac{1-b}{m}   \cd (\Human - \Alg) - (\Human - \Alg)^2} <  & (\delta_h - \delta_a)^2 
\end{align*}
where we have used the assumption that $0 < p < 1$ and $m>0$. If the lefthand side of the inequality is positive, we can take the square root of both sides to get: 
$$\sqrt{\Human - \Alg} \cd \sqrt{\frac{1-\prob}{\prob}  \p{\frac{1-b}{m} - \p{\Human - \Alg} }}<  \vert \delta_h - \delta_a \vert $$
Finally, we will show that the term under the square root, give the assumptions of this lemma. The term under the square root is \emph{negative} if: 
$$\frac{1-b}{m}  - \p{\Human - \Alg} <0$$
or:
\begin{equation}\label{eq:sqrteq}
1 <  b - m \cd (\Human - \Alg)
\end{equation}
We will show that, if this happens, we must have $\select(\alg_1, \human_1) > 1$ or $\select(\alg_2, \human_2)>0$ (either of which violate the assumptions of this lemma). 
For this combining rule, 
$$\select(\alg_1, \human_1) = b - m \cd (\human_1 - \alg_1) = b - m \cd (\Human + \delta_h - \Alg - \delta_a) = b - m \cd (\Human - \Alg) + m \cd (\delta_h- \delta_a)$$
In the event that $\delta_h > \delta_a$, then the above equation is greater than Equation \ref{eq:sqrteq}. Therefore, if Equation \ref{eq:sqrteq} is greater than 1, then $\select(\alg_1, \human_1)$ is also greater than 1, which violates the assumptions of the lemma. 
Similarly, we can write: 
$$\select(\alg_2, \human_2) = b - m \cd (\human_2 - \alg_2) = b - m \cd \p{\Human - \delta_h \cd \frac{\prob}{1-\prob} - \Alg + \delta_a \cd \frac{\prob}{1-\prob}} = b - m \cd (\Human - \Alg) + m \cd \frac{\prob}{1-\prob} \cd (\delta_a- \delta_h)$$
In the event that $\delta_h < \delta_a$, then the above equation is greater than Equation \ref{eq:sqrteq}. Therefore, if Equation \ref{eq:sqrteq} is greater than 1, then $\select(\alg_2, \human_2)$ is also greater than 1, which violates the assumptions of the lemma. 
\end{proof}
\end{document}